\documentclass[conference, twocolumn]{IEEEtran}
\usepackage{mdwlist}
\usepackage{cite}

\usepackage{amsmath}
\usepackage{amssymb}
\usepackage{amsthm}

\usepackage[final]{graphicx}
\usepackage{color}
\usepackage{subfigure}
\usepackage{epstopdf}

\usepackage{url}

\newtheorem{theorem}{Theorem}[section]

\newtheorem{lemma}[theorem]{Lemma}
\newtheorem{corollary}[theorem]{Corollary}

\theoremstyle{definition}

\theoremstyle{remark}

\title{Two Species Evolutionary Game Model of User and Moderator Dynamics}
\author{
\IEEEauthorblockN{Christopher Griffin 
and Douglas Mercer}
\IEEEauthorblockA{Applied Resarch Laboratory and \\
Department of Mathematics\\
Penn State University\\
University Park, PA 16802\\
E-mail: \texttt{griffinch@ieee.org}}
\and
\IEEEauthorblockN{James Fan}
\IEEEauthorblockA{Dept. of Supply Chain and Logistics\\
Penn State University\\
University Park, PA 16802\\
E-mail: \texttt{juf187@psu.edu}}
\and\centering
\IEEEauthorblockN{Anna Squicciarini}
\IEEEauthorblockA{College of Info. Sci. and Tech.\\
Penn State University\\
University Park, PA 16802\\
E-mail: \texttt{asquicciarini@psu.edu}}
}

\begin{document}
\maketitle
 
\begin{abstract}
We construct a two species evolutionary game model of an online society consisting of ordinary users and behavior enforcers (moderators). Among themselves, moderators play a coordination game choosing between being ``positive'' or ''negative'' (or harsh) while ordinary users play prisoner's dilemma. When interacting, moderators motivate good behavior (cooperation) among the users through punitive actions while the moderators themselves are encouraged or discouraged in their strategic choice by these interactions. We show the following results: (i) We show that the $\omega$-limit set of the proposed system is sensitive both to the degree of punishment and the proportion of moderators in closed form. (ii) We demonstrate that the basin of attraction for the Pareto optimal strategy $(\text{Cooperate},\text{Positive})$ can be computed exactly. (iii) We demonstrate that for certain initial conditions the system is self-regulating. These results partially explain the stability of many online users communities such as Reddit. We illustrate our results with examples from this online system.
\end{abstract}

\section{Introduction}
Online social networks (OSNs) allow individuals to share not only information about themselves, but also information about their world. While traditional OSNs, such as Facebook, Myspace, and Google+, are designed for sharing personal information, new OSNs  with a focus on sharing all interesting information, personal or not, have become popular.  These OSNs, such as Reddit, Digg, and Funnyjunk, still allow users to create online connections through shared interests and experiences.  Much like traditional OSNs, these new OSNs are susceptible to user deception.

Reddit in particular rewards users for posting information with a form of social capital called karma.  Shared information, in the form of posts, is upvoted or downvoted by the community.  The user in turn accumulates karma based on how popular his or her posts are. While karma has no intrinsic value, it appears to be a desirable commodity within the Reddit community. Karma is a unique and concrete measure of social capital, a great tool to analyze online social behavior. One interesting component of Reddit is the communal agreement that honesty is highly preferred to dishonesty within the system, should be rewarded, and moderators actively enforce this policy through a variety of methods. In this paper, we seek to develop dynamics that lead to the unique properties found on Reddit: that is a self-stabilizing society in which moderators (or authority figures) act in as kind a way as possible with users committed to cooperation through honesty at the potential expense of karma.

\subsection{An Overview of Reddit}
From the perspective of users, lying without malice does not necessarily affect the holistic well-being of the Reddit social network. Instead, users identify ``trolls'' as problematic. A troll is a member of an online community whose contributions are intended to enrage or offend as many people as possible, as significantly as possible. In essence, trolls seek out opportunities to defend or propose indefensible, reprehensible positions in order to receive a negative reaction from other community members.

Many Reddit users will ignore trolls. However, in cases where a troll is not ignored, online altercations may occur leading to communal breakdown within Reddit. Beyond the use of karmic voting and individual interactions, Reddit enforces community standards by utilizing two layers of moderators. Specifically, subreddit\footnote{A subreddit is a specific discussion forum within Reddit} specific moderators enforce Reddit and subreddit guidelines by deleting posts and comments that were not already buried by downvotes. These are the moderators upon which the paper focuses. If these moderators are not strict enough, subreddits sympathetic to trolling can harbor and encourage trolling behavior, disrupting the community.

In this paper we model negative behavior within Reddit (or a similar online community) as the defect behavior in a classic prisoner's dilemma. We justify this assumption by noting that trolling is analogous to defecting, in the sense that defection is characterized by an action that goes against what is best for societal wellbeing for selfish reasons. In Section IV, where we explictly define our user-user payoff matrix, we will detail our justification for this assumption in further detail. 

\subsection{Paper Summary}
In this paper, we develop a simple evolutionary game that attempts to model the behavior of Reddit users and moderators and illustrates how the Reddit equilibrium can be reached; that is, an equilibrium in which most users agree to cooperatively share information and in which moderators beneficially interact with the system. We assume that ordinary users interact with each other playing a prisoner's dilemma style game, while moderators interact with each other playing a coordination game. The moderators' strategy space consists of the strategies, ``Positive'' and ``Negative,'' which attempts to capture their view of users, in particular when they are engaged in negative behavior. When a moderator and a user interact, moderators may or may not derive benefit from the interaction depending on the strategy of the user. A similar statement holds for the user.

Specifically, given our evolutionary game system discussed in the sequel, we show the following results: (i) We show that the $\omega$-limit set of the proposed system is sensitive both to the degree of punishment and the proportion of moderators in closed form. (ii) We demonstrate that the basin of attraction for the Pareto optimal strategy $(\text{Cooperate},\text{Positive})$ can be computed exactly. (ii) We demonstrate that for certain initial conditions the system is a regulating. These results partially explain the stability of many online users communities such as Reddit.

The remainder of this paper is organized as follows: In Section \ref{sec:LitRev}, we provide a brief literature review of deviance (or negative behavior) in online social environments and a review of the impact moderators have in this situation. In Section \ref{sec:BigModel} we layout our modeling approach and contrast it to the one in \cite{Hof96}. In Section \ref{sec:Model} we present our basic model of the system. In Section \ref{sec:EvolutionaryGame} we present our results on the dynamical system under consideration. We present future directions and conclusions in Section \ref{sec:Conclusion}.

\section{Literature Review}\label{sec:LitRev}
 Online deviance (defection, in our model) has been studied from the perspective of both social science and computer science. To date these two perspectives have not been adequately integrated. Using classical labeling and identity theories, several social researchers have focused on Internet users' behavior
 \cite{davis,17,18,20,22}. Labeling theory holds that being labeled as a ``deviant" leads a person to engage in deviant behavior, and explains why people's behavior clashes with social norms  \cite{davis}. However, social sciences research has not yet developed a normative definition of cyber communication and the online subculture. A significant study of online behavior that will inform our model is the ``Palace Study'' in which Suler and Philips \cite{18} classify deviant behavior (strategies) into several types and provide a taxonomy of possible counter-strategies, ranging from mild, premeditative actions (e.g., warnings) to preventative systems. The study confirms that existing prevention and remediation of aversive online behavior techniques have been found to be difficult and expensive. For instance, reputation systems were found to be useful for this purpose, but are unreliable due to the lack of identity validation and control \cite{six}. Recent socio-computational studies focus mostly on single short text analysis to automatically identify spam/deviant comments in user-contributed sites \cite{Wang:2010,44,40,30,37}. While content-based methods have shown encouraging results, they are limited to single-post analysis, and  they do not target specific users' behaviors or follow traces. 

Several tools exist to help moderators identify bots and vandalism (e.g. \cite{ot2,onlinetool}). 
Automated bots (e.g., Cluebot), filters (e.g., abusefilter), and 
editing assistants (e.g., Huggle and Twinkle) all aim to 
locate acts of vandalism. Such tools work via 
regular expressions and manually-authored rule sets. 
In addition, a notable effort is from  West  and colleagues \cite{West:2010}, who  adopted classifiers to detect   vandalism on Wikipedia. 
At the core of the West's solution is a lightweight classifier capable of identifying vandalism. The classifier exploits temporal and spatial features,  extracted from 
revision metadata of articles.

Our work also parallels the body of work on free-riding in peer-to-peer systems \cite{feldman2006}. Peer-to-peer systems are designed to allow users to connect with others and share resources. Similar to
online communities, users are free to access and contribute as much as desired, and few controls are in place.  As for online communities, punishments, although applied, are shown not to be truly effective, most likely because users can abandon the system.  To tackle these issues,  the common solution is to  implement incentive-based mechanisms.  Incentives are applied  in certain online forums, whereby end users are given special roles and privileges as a result of their good-standing (see \cite{JCC4:JCC401} for a discussion on the community enforcement mechanisms of E-Bay). In this paper, we assume that users cannot easily abandon the system (i.e., there are few competitors and a barrier to change, as there is with Facebook) and study the case when moderators focus on punishments rather than incentives. (See Section \ref{sec:Model}.) We discuss how to vary the model to study the incentives based case in Section \ref{sec:Conclusion}.

\section{Relevant Previous Work in Evolutionary Games}\label{sec:BigModel}
Consider a two-player bimatrix game. That is, the payoff matrix for the row player is $\mathbf{A} \in \mathbb{R}^{n\times n}$ ($n \in \mathbb{Z}_+$) and for the column player it is $\mathbf{B}$. In an evolutionary game, let $\boldsymbol{\zeta}(t) \in \mathbb{R}^{n \times 1}$ be a vector whose $i^\text{th}$ component $\boldsymbol{\zeta}_i(t)$ yields the proportion of the population of \textit{row players} that chooses pure strategy $i$ at time $t$. We will likewise define $\boldsymbol{\chi}(t) \in \mathbb{R}^{n \times 1}$ for the column players. Hofbauer \cite{Hof96} proposes the following replicator dynamics for this case:
\begin{gather}
\dot{\boldsymbol{\zeta}}_i = \boldsymbol{\zeta}_i\left( \left(\mathbf{A}\boldsymbol{\chi}\right)_i - \boldsymbol{\zeta}^T\mathbf{A}\boldsymbol{\chi} \right) \quad i = 1,\dots, n
\label{eqn:Hof1}
\\
\dot{\boldsymbol{\chi}}_i = \boldsymbol{\chi}_j \left(\left(\boldsymbol{\zeta}^T\mathbf{B}\right)_j - \boldsymbol{\zeta}^T\mathbf{B}\boldsymbol{\chi}\right) \quad j = 1,\dots,n
\label{eqn:Hof2}
\end{gather}
This is a simple generalization of the replicator dynamics from a zero-sum  game to a general sum game. 

For games in which a population is to play a symmetric bimatrix game, we propose the following simplified dynamics. Let $\boldsymbol{\zeta}(t) \in\mathbb{R}^{n\times n}$ simply be the vector whose $i^\text{th}$ component is the proportion of the population that is playing pure strategy $i$ at time $t$. Then for an individual playing strategy $i$, the expected payoff  value is nothing more than $(\mathbf{A}\boldsymbol{\zeta})_i$, regardless of whether this individual is a row or column player since $\boldsymbol{\zeta}^T\mathbf{A^T} = \mathbf{A}\boldsymbol{\zeta}$ by symmetry. Care must be taken, however, when computing the population average. In this case, the population average is not $\boldsymbol{\zeta}^T\mathbf{A}\boldsymbol{\zeta}$ as it is in the case of the classical replicator dynamics \cite{Wei95}. Instead, the population average is given by:
\begin{equation}
\bar{u} = \frac{1}{2}\boldsymbol{\zeta}^T\left(\mathbf{A} + \mathbf{A}^T\right)\boldsymbol{\zeta}
\label{eqn:PopulationAverage}
\end{equation}
To see this, assume that (as expected) half the time a player meets a competitor she will play the role of the row player and the other half of the time she will play the roll of the column player. Then the population average can be computed as:
\begin{equation}
\bar{u} = \frac{\boldsymbol{\zeta}^T\mathbf{A}\boldsymbol{\zeta} + \boldsymbol{\zeta}^T\mathbf{A}^T\boldsymbol{\zeta}}{2}
\end{equation}
which is identical to Equation \ref{eqn:PopulationAverage}. This leads to a simplified replicator dynamic in the case of a symmetric game:
\begin{equation}
\dot{\boldsymbol{\zeta}}_i = \boldsymbol{\zeta}_i\left( \left(\mathbf{A}\boldsymbol{\zeta} \right)_i - \frac{1}{2}\boldsymbol{\zeta}^T\left(\mathbf{A} + \mathbf{A}^T\right)\boldsymbol{\zeta} \right)
\label{eqn:SimpleRepDyn}
\end{equation} 
We will use this dynamic for the evolution of strategy \textit{within} our subpopulations of ordinary users and moderators while we will use the formulation of Hofbauer in our inter-population strategy evolution dynamics.

\section{Model} \label{sec:Model}
Let $x(t) \in [0,1]$ for all $t \in \mathbb{R}_+$ be the proportion of ordinary users who choose to cooperate (e.g., behave appropriately, act honestly, etc.) while $y(t) \in [0,1]$ for all $t \in \mathbb{R}_+$ is the proportion of users who choose to defect (e.g., behave negatively, deceive, etc.). Naturally our dynamics will require $x(t) + y(t) = 1$ for all $t \in \mathbb{R}_+$. Likewise, let $z(t)$ be the proportion of moderators who choose to be positive and let $w(t)$ be the proportion of moderators who choose to be negative so that $z(t) + w(t) = 1$ for all $t \in \mathbb{R}_+$ as well. 

When interacting, each subpopulation plays a symmetric general sum game: prisoner's dilemma (\cite{Gri12}, Page 67) and a coordination game respectively. In Prisoner's Dilemma, users who cooperate gain a benefit, but not so much as a user who defects from a cooperating user. Two defecting users reap a smaller benefit than they would if they cooperate.

For ordinary users, assume they have the following prisoner's dilemma payoff matrix:
\begin{equation}
\mathbf{A} = \begin{bmatrix} \frac{r}{2} & -r\\r & \frac{r}{4}
\end{bmatrix}
\end{equation}

This assumption follows naturally from the following obersations. Consider an interaction between two OSN users. If both choose to cooperate, we would expect their interaction to be beneficial to the community as a whole (e.g., produce an insightful, genuine conversation). However, if one user chooses to defect while the other opts to abide by community standards, the defecting player will derive his satisfaction at the expense of the unsuspecting, cooperating user (e.g., the troll reaping the reward of anger from the legitimate user). If both users defect, each will derive some satisfaction from the experience (e.g., a shared joke between two like-minded defectors), though not as much as if they had cooperated. Thus, the dominant strategy is, as in the case of the textbook prisoner's dilemma problem, to defect, despite the social optimum strategy being user cooperation.

In a coordination game, users who play the same strategy are rewarded, while users who do not are penalized. The moderators have the following coordination game payoff matrix:
\begin{equation}
\mathbf{F} = \begin{bmatrix} v & -v\\-v & v
\end{bmatrix}
\end{equation}
For simplicity in this paper, we will assume that $r = v = 1$ and leave results on the more general case to subsequent work. When the two populations interact, they play a bimatrix with payoff matrices given as:
\begin{gather}
\mathbf{B} = \begin{bmatrix}\frac{a}{2} & 0 \\ -\frac{a}{2} & -a\end{bmatrix}\\
\mathbf{C} = \begin{bmatrix}\frac{s}{2} & 0 \\ \frac{s}{4} & s\end{bmatrix}
\end{gather}
Here $\mathbf{B}$ is the payoff matrix for ordinary users (as the row players) and $\mathbf{C}$ is the payoff matrix for moderators (as the column players). It is easy to see for $a,s>0$ that ordinary users benefit from meeting a positive moderator when they are cooperating and gain nothing when they meet a negative moderator. When an ordinary user is defecting he is penalized when he meets any member of the moderator subpopulation, but more so when he meets a negative moderator. Likewise, moderators acting positively benefit when they meet any player, but less so when they meet a defector (presumably it makes them unhappy to consider a user engaged in negative behavior, but they are able to moderate behavior, thus ``improving society''). Moderators who are negative derive no pleasure from meeting a cooperating user, but substantial pleasure from punishing (or expelling) a defecting user. For the sake of simplicity, we will assume that $s = 1$ for the remainder of this paper. While these are specific payoff matrices, we assert that the qualitative behavior we observe will be largely the same no matter how we assign values, even in the presence of a more complex game structure. Essentially, as long as the users are playing prisoner's dilemma, the moderators are playing a coordination game, and there is a penalty when a defector meets a moderator, then the qualitative behaviors we observe will be present.

Assume that in a population of players a proportion $n_p \in (0,1)$ are ordinary users and $n_c  \in (0,1)$ are moderators where $n_p + n_c = 1$. Thus, $n_p$ is the proportion of the population that is an ordinary user. Let:
\begin{equation}
\boldsymbol{\xi}(t) = \begin{bmatrix}x(t)\\y(t)\end{bmatrix} \quad 
\boldsymbol{\eta}(t) = \begin{bmatrix}z(t)\\w(t)\end{bmatrix}
\end{equation}
Combining the dynamics given in Equation \ref{eqn:SimpleRepDyn} with Hofbauer's dynamics (Equations \ref{eqn:Hof1} and \ref{eqn:Hof2}), we obtain the following dynamics for the player:
\begin{multline}
\dot{\boldsymbol{\xi}}_i = 
\boldsymbol{\xi}_i \left(n_p\left( \left(\mathbf{A}\boldsymbol{\xi} \right)_i - \frac{1}{2}\boldsymbol{\xi}^T\left(\mathbf{A} + \mathbf{A}^T\right)\boldsymbol{\xi} \right) + \right.\\
\left. n_c  \left( \left(\mathbf{B}\boldsymbol{\eta}\right)_i - \boldsymbol{\xi}^T\mathbf{B}\boldsymbol{\eta} \right)\right) \quad i = 1,2
\label{eqn:DynSys1}
\end{multline}

\begin{multline}
\dot{\boldsymbol{\eta}}_j = 
\boldsymbol{\eta}_j \left(n_c\left( \left(\mathbf{F}\boldsymbol{\eta} \right)_j - \frac{1}{2}\boldsymbol{\eta}^T\left(\mathbf{F} + \mathbf{F}^T\right)\boldsymbol{\eta} \right) + \right.\\
\left. n_p  \left( \left(\boldsymbol{\xi}^T\mathbf{C}\right)_j - \boldsymbol{\xi}^T\mathbf{C}\boldsymbol{\eta} \right)\right) \quad j = 1,2
\label{eqn:DynSys2}
\end{multline}

In the sequel, we will explore the dynamics of stability for varying values of $n_p$ and $a$, the relative punitive value of defecting when playing against a member of the moderators subpopulation. 

We first state a theorem that will simplify our analysis of these dynamics. Essentially, it simply asserts we can solve these differential equations completely on the subspace $x(t) + y(t) = 1$ and $z(t) + w(t) = 1$.

\begin{theorem} On the subspace defined by the equalities $x(t) + y(t) = 1$ and $z(t) + w(t) = 1$, the dynamical system given in Equations \ref{eqn:DynSys1} and \ref{eqn:DynSys2} is equivalent to the two-variable differential system:
\begin{gather}
\dot{x} = \frac{1}{4}\,x \left( -1+x \right)  \left( -3{\it n_p}\,x+5\,{\it n_p}+4\,{\it n_p}\,a-4\,a \right)\label{eqn:DynSysA1}\\
\dot{z} = -\frac{1}{4}z \left( -1+z \right)  \left( -16{\it n_p}\,z+16\,z+5\,{\it n_p}
-8+5{\it n_p}\,x \right) 
\label{eqn:DynSysA2}
\end{gather}
\end{theorem}

\section{Theoretical Results}\label{sec:EvolutionaryGame}
The following theorem on the equilibria of the dynamical system given by Equations \ref{eqn:DynSys1} and \ref{eqn:DynSys2} is easily verified by substitution. We note there are nine equilibria that can be identified by finding the roots of the right hand sides of Equations \ref{eqn:DynSysA1} and \ref{eqn:DynSysA2}.

\begin{theorem} For the dynamical system given by Equations \ref{eqn:DynSysA1} and \ref{eqn:DynSysA2}, there are always  9 equilibria, (some possibly spurious): 
\begin{enumerate*}
\item $x = 0, z = 0$
\item $x = 0,z = 1$
\item $x = 1, z = 0$
\item $x = 1, z = 1$
\item $ x=0, z=\tfrac{1}{16}{\tfrac {5{\it n_p}-8}{-1+{\it n_p}}}$ 
\item $x=\tfrac{1}{3}{\tfrac {5{\it n_p}+4{\it n_p}\,a-4\,a}{{\it n_p}}},z=0$
\item $x=1, z=\tfrac{1}{8}{\tfrac {5{\it n_p}-4}{-1+{\it n_p}}}$
\item $x=\tfrac{1}{3}{\tfrac {5{\it n_p}+4\,{\it n_p}\,a-4\,a}{{\it n_p}}},z=
1$
\item $x=\tfrac{1}{3}{\tfrac {5{\it n_p}+4\,{\it n_p}\,a-4\,a}{{\it n_p}}}, z=
\tfrac{1}{12}{\tfrac {10\,{\it n_p}-6+5\,{\it n_p}\,a-5\,a}{-1+{\it n_p}}}$
\end{enumerate*}
\label{thm:Equil}
\end{theorem}
It is worthwhile noting that these equilibria may not always be valid for our equations. It may be that the equilibrium points fall outside the solution space $x \in [0,1]$ and $z \in [0,1]$. In this case, these stationary points are spurious. Of interest is the ninth equilibrium point, because it is an interior equilibrium point:
\begin{corollary} Assume $n_p \in (0,1)$, then there is a non-trivial, non-spurious equilibrium point for which $x,z \in (0,1)$ (and thus $y,w \in (0,1)$) just in case:
\begin{displaymath}
\begin{cases}
\frac{1}{2}\,{\frac {{\it n_p}}{1-{\it np}}} < a < \frac{5}{4}\,{\frac {{\it n_p}}{1-{\it n_p}}} & \text{if}\,\, n_p < \frac{8}{11}\\

\frac{1}{2}\,{\frac {{\it n_p}}{1-{\it np}}} < a < \frac{2}{5}\,{\frac {3-{\it n_p}}{1-{\it np}}} & \text{if}\,\, \frac{8}{11} \leq n_p < \frac{4}{5}\\

\frac{2}{5}\,{\frac {5\,{\it n_p}-3}{1-{\it np}}} < a < \frac{2}{5}\,{\frac {3-{\it n_p}}{1-{\it np}}} & \frac{4}{5} \leq n_p
\end{cases}
\end{displaymath} 
\label{thm:Interior}
\end{corollary}

A more interesting question revolves around the stability of the various equilibrium points. It would be nice to know that the presence of the moderator subpopulation causes the cooperate strategy to become stable within the ordinary users, even though in general prisoner's dilemma it is not in any player's interest to cooperate. Moreover, we would also like a society in which the moderators play the positive strategy, since there's no point in living in a society where users behave because they are terrified of their system of justice. (Presumably, the online society would fall apart.) We can explore this problem by computing the eigenvalues of the Jacobian matrix of the dynamical system. That is, by studying the characteristics of the non-linear dynamical system described in Equations \ref{eqn:DynSys1} and \ref{eqn:DynSys2} by linearizing about a stable point of interest. In the following lemma, we linearize about the utopian equilibrium point $x=1,z=1$ and use the eigenvalues of the Jacobian to determine when this point is stable.
\begin{lemma} The Jacobian matrix $\mathbf{H}$ of the dynamical system described by Equations \ref{eqn:DynSysA1} and \ref{eqn:DynSysA2} about $x=1,z=1$ (and $y = 0$ and $w = 0$) is specified by:
\begin{equation}
\mathbf{H} = \left[ \begin {array}{cc} 1/2\,{\it n_p}-a+{\it n_p}\,a&0
\\ \noalign{\medskip}0&-2+3/2\,{\it n_p}\end {array} \right]
\end{equation}
with eigenvalues:
\begin{displaymath}
 \left[ \begin {array}{c} -2+3/2\,{\it n_p}\\ \noalign{\medskip}1/2\,{
\it n_p}-a+{\it n_p}\,a\end {array} \right]
\end{displaymath}
\end{lemma}

From the previous lemma and Theorem 3.2 of \cite{Verh06} the following theorem is immediate:
\begin{theorem} Assume $n_p \in (0,1)$. If:
\begin{displaymath}
\frac{n_p}{2(1-n_p)} < a 
\end{displaymath}
then $x=1,z=1$ (and $y = 0$ and $w = 0$) is a stable equilibrium point.
\label{thm:StableEq}
\end{theorem}

The resulting relationship between the two variables is illustrated in Figure \ref{fig:npvsa}.
\begin{figure}[htbp]
\centering
\includegraphics[scale=0.3]{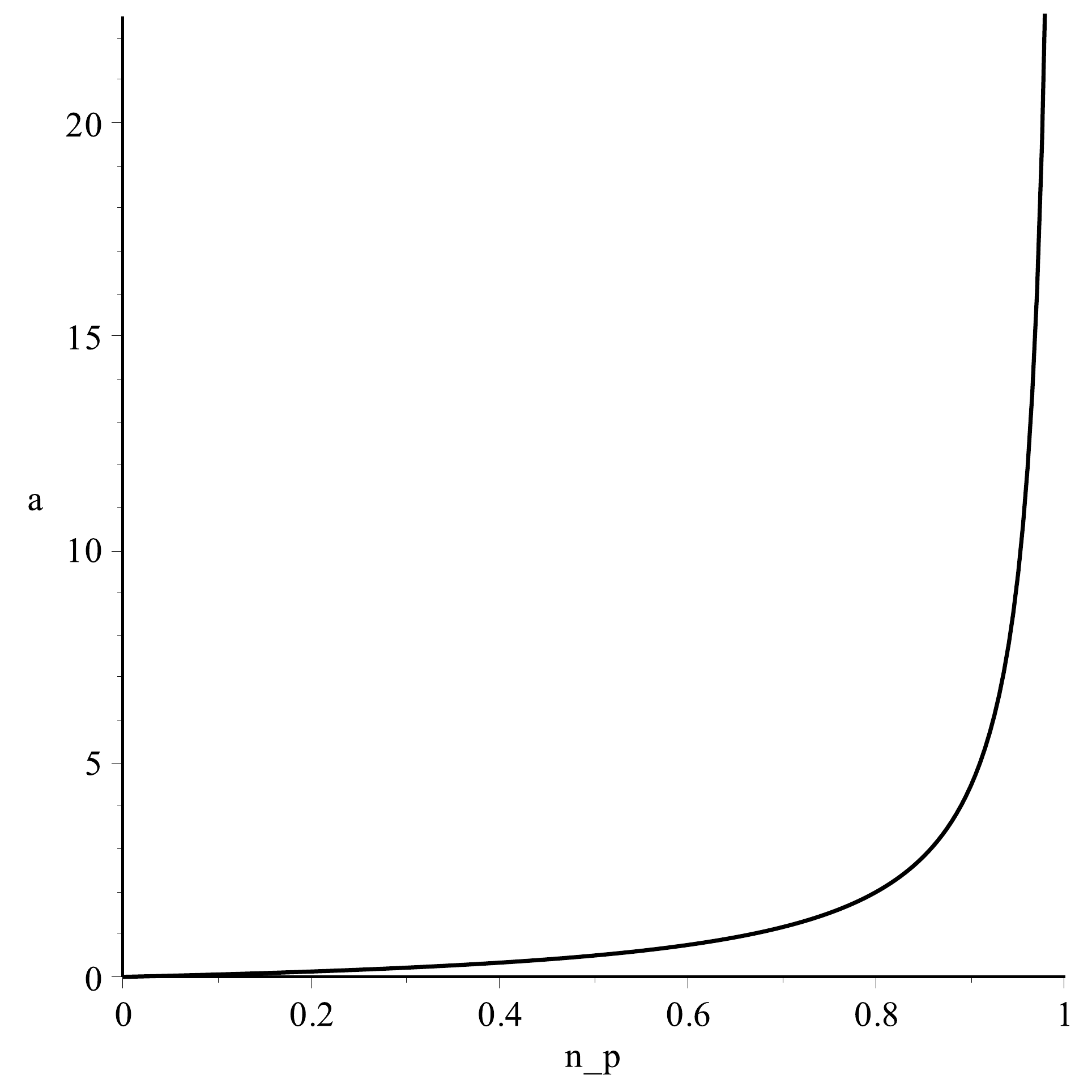}
\caption{The relationship between $n_p$ and $a$ when we ensure that $[x=1,y=0,z=1,w=0]$ is a stable equilibrium point.}
\label{fig:npvsa}
\end{figure}
This figure makes a great deal of sense. As the proportion of the population becomes overwhelmingly dominated by ordinary users, the probability of encountering a member of the moderators subpopulation drops. Therefore, to ensure proper behavior, stricter and stricter punitive action is required. 

By a similar process, we can also explore the case of the dystopian society in which $[x=0,y=1,z=0,w=1]$ is stable. 
\begin{theorem} Assume $n_p \in (0,1)$. If:
\begin{displaymath}
\frac{5n_p}{4(1-n_p)} > a 
\end{displaymath}
then $x=0,z=0$ (and $y = 1$ and $w = 1$) is a stable equilibrium point.
\label{thm:UnstableEq}
\end{theorem}
\begin{corollary} There is at least one pair of values for $a$ and $n_p$ so that both $[x=0,y=1,z=0,w=1]$ and $[x=1,y=0,z=1,w=0]$ are stable.
\label{cor:DualStability}
\end{corollary}
Corollary \ref{cor:DualStability} tells us that in our online system, it is possible to ``descend into chaos'' in the sense that all users are actively defecting (deceiving, scamming etc.) and all moderators are engaged in highly punitive activities. Corollary \ref{cor:DualStability} also suggests that an investigation of the basins of attraction for the two attracting points could lead to a complete characterization of the behavior of the dynamical system in light of the following theorem, which follows from Theorem 3.1 of \cite{Verh06}:
\begin{theorem} Under no conditions is the interior equilibrium point 
\begin{gather*}
x=\tfrac{1}{3}{\tfrac {5{\it n_p}+4\,{\it n_p}\,a-4\,a}{{\it n_p}}}\\
z=\tfrac{1}{12}{\tfrac {10\,{\it n_p}-6+5\,{\it n_p}\,a-5\,a}{-1+{\it n_p}}}
\end{gather*}
ever stable.
\end{theorem}
\begin{proof} Analysis of Jacobian matrix and eigenvalues shows that for this point to be stable, we must have:
\begin{gather*}
a < \frac{1}{2}\,{\frac {{\it n_p}}{1-{\it n_p}}} \text{ or } \frac{5}{4}\,{\frac {{\it n_p}}{1-{\it n_p}}} < a \text{ and }\\
a<\frac{2}{5}\,{\frac {5\,{\it n_p}-3}{1-{\it n_p}}} \text{ or } \frac{2}{5}\,{\frac {3-{\it n_p}}{1-{\it n_p}}}<a
\end{gather*}
From Theorem \ref{thm:Interior}, for the equilibrium to be non-spurious (i.e., in $[0,1] \times [0,1]$) it must be the case that:
\begin{displaymath}
a < \frac{1}{2}\,{\frac {{\it n_p}}{1-{\it n_p}}} \text{ and }
\frac{2}{5}\,{\frac {3-{\it n_p}}{1-{\it n_p}}}<a
\end{displaymath}
If these intervals overlap, then there is a point at which:
\begin{displaymath}
\frac{n_p}{2} = \frac{2}{5}(3 - n_p)
\end{displaymath}
which only occurs if $n_p = \tfrac{4}{3}$, but we know $n_p \in (0,1)$. Thus, the interior equilibrium is always unstable.
\end{proof}
What the preceding theorem means is that for many online systems that obey the dynamics discussed will either converge to a final state in which all users are behaving cooperatively and moderators who are positive or it will descend into chaos. This is illustrated in the figures below and proved explicitly for certain $n_p$.

To illustrate the nature of the equilibria, we consider the case when $n_p = 0.9$ and $a = 7$. (In this case, 90\% of the population is composed of ordinary users.) From Theorem \ref{thm:Equil}, we can see that the non-trivial interior point equilibrium is present with values: $x = \tfrac{17}{27}$ and $z = \tfrac{5}{12}$ as are 6 other equilibrium points. The system phase portrait is shown in Figure \ref{fig:Cliff1}:
\begin{figure}[htbp]
\centering
\includegraphics[scale=0.38]{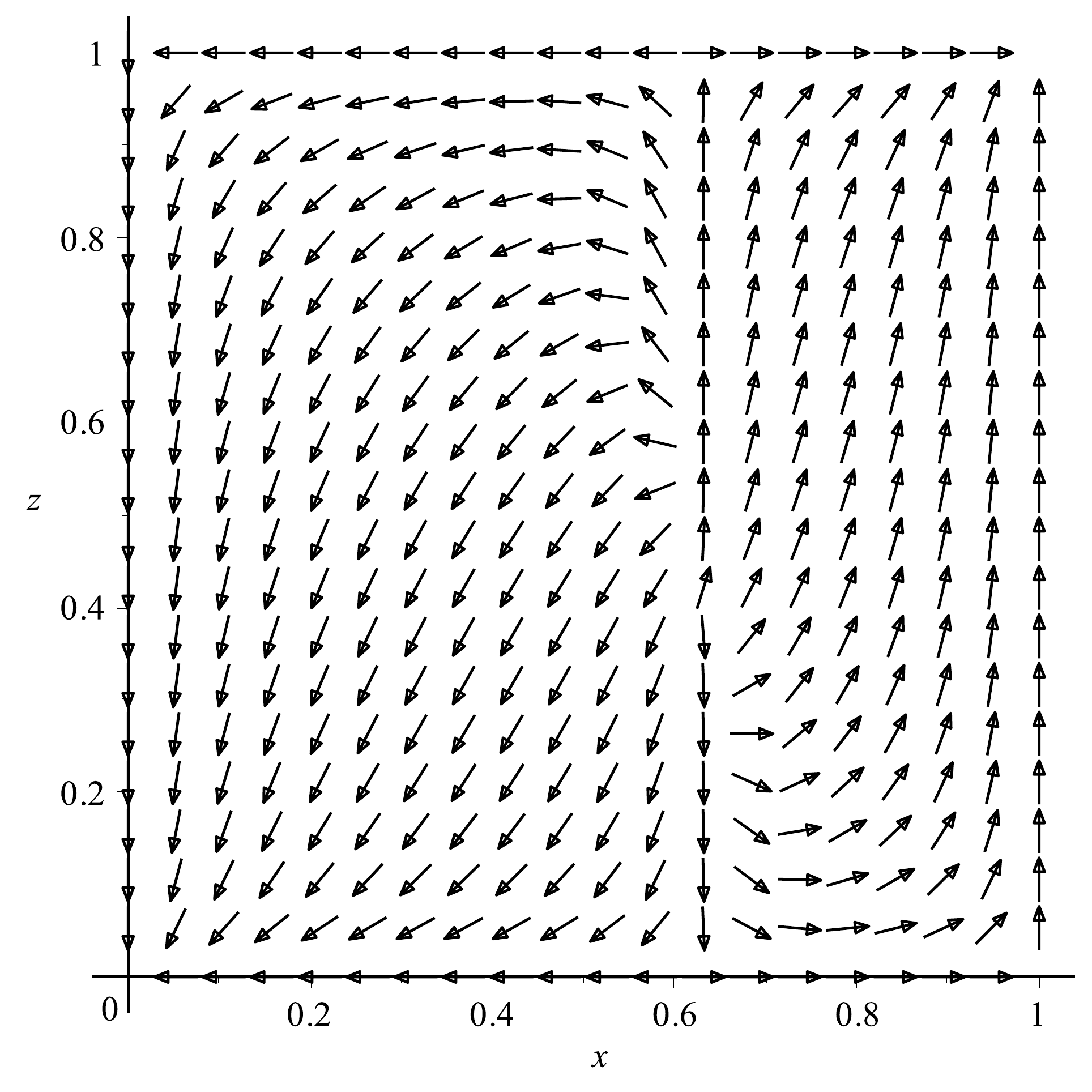}
\caption{Phase portrait of the dynamical system when $a = 7$ and $n_p=0.9$}
\label{fig:Cliff1}
\end{figure} 
In a case like this, we can see that the space half-space $x<\tfrac{17}{27}$, $z \in (0,1)$ is the basin of attraction for the point for $x=0,z=0$ (except for a set of measure zero) while the half-space $x>\tfrac{17}{27}$, $z \in (0,1)$ is the basin of attraction for $x=1,z=1$ (again except for a set of measure zero).

The behavior of the population varies substantially with the value of $a$ and the previously illustrated behavior is not the only possible outcome for this online society. An interesting situation arises when we set $n_p = 0.9$ and $a = 12$. In this case, there is no interior equilibrium point and the basin of attraction for $x=1,z=1$ is \textit{almost} the entire region $[0,1] \times [0,1]$. This is shown in Figure \ref{fig:Cliff2}.
\begin{figure}[htbp]
\centering
\includegraphics[scale=0.38]{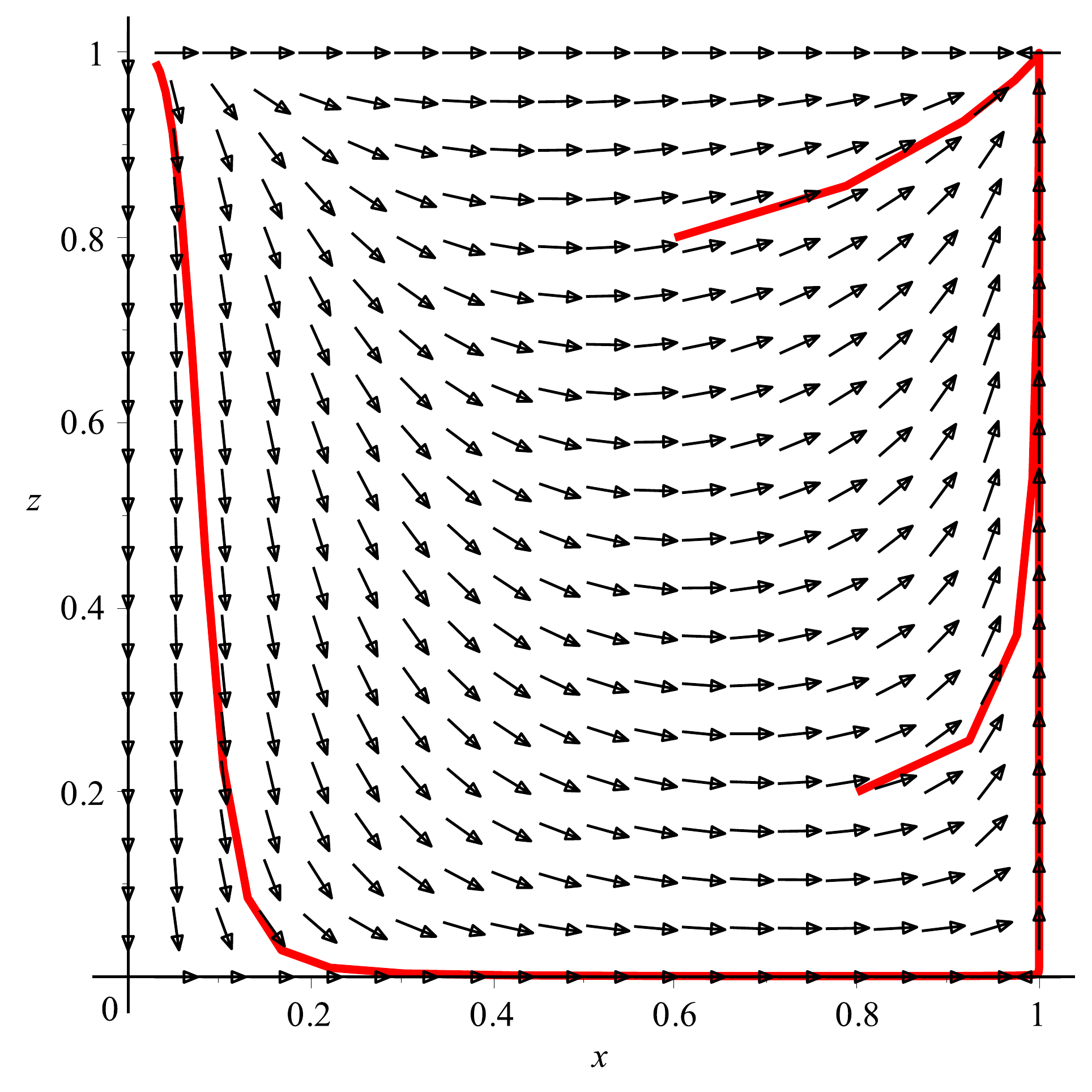}
\caption{Phase portrait of the dynamical system when $a = 12$ and $n_p=0.9$}
\label{fig:Cliff2}
\end{figure} 
What's interesting about this case is the self-regulating nature of the system. Note the trajectory beginning at $x = \tfrac{3}{100}$ and $z = \tfrac{99}{100}$. We see that the moderators, while starting with the positive strategy, quickly change to the negative strategy, which drives the ordinary users to move from the defect strategy to the cooperate strategy. This, in turn, drives the moderators to move from the negative strategy to the positive strategy, arriving in the utopian scenario. We now provide a result on the basin of attraction in scenarios like the previous example. This result shows that the behavior illustrated in Figure \ref{fig:Cliff1} is somewhat typical of this system and describes how one online site can become successful and (mostly) stable like Reddit, while other sites might descend into chaos and fail.

Suppose that there is a non-trivial equilibrium solution for which $x^* \in (0,1)$, that is:
\begin{equation}
x^*=\tfrac{1}{3}{\tfrac {5{\it n_p}+4\,{\it n_p}\,a-4\,a}{{\it n_p}}}
\label{eqn:xstar}
\end{equation} 
Consider $x^* - \epsilon$ where $\epsilon \in (0,x^*)$. If we evaluate 
\begin{displaymath}
\dot{x} = \frac{1}{4}\,x \left( -1+x \right)  \left( -3{\it n_p}\,x+5\,{\it n_p}+4\,{\it n_p}\,a-4\,a \right)
\end{displaymath}
at $x^*-\epsilon$, then we obtain:
\begin{equation}
\frac{1}{12}\frac{r\epsilon s}{n_p}
\label{eqn:EvalXdot}
\end{equation}
where
\begin{gather*}
r = \left( 5\,{\it n_p}-4\,a+4\,{\it n_p}\,a-3\,\epsilon\,{
\it n_p} \right)\\
s = \left( 4\,{\it n_p}\,a-4\,a+2\,{\it n_p}-3\,
\epsilon\,{\it n_p} \right)
\end{gather*}
Expression \ref{eqn:EvalXdot} is cubic in $\epsilon$ and it has three roots: 
\begin{displaymath}
\left\{0,\frac{2}{3}\,{\frac {2\,{\it n_p}\,a-2\,a+{\it n_p}}{{\it n_p}}},
\frac{1}{3}\,{\frac {
5\,{\it n_p}-4\,a+4\,{\it n_p}\,a}{{\it n_p}}}\right\}
\end{displaymath}
We can see at once that:
\begin{equation}
\frac{1}{3}\,{\frac {
5\,{\it n_p}-4\,a+4\,{\it n_p}\,a}{{\it n_p}}} - \frac{2}{3}\,{\frac {2\,{\it n_p}\,a-2\,a+{\it n_p}}{{\it n_p}}} = 1
\end{equation}
and, by our assumption in Equation \ref{eqn:xstar}, the Expression \ref{eqn:EvalXdot} is either always positive or always negative on the interval:
\begin{displaymath}
\left[0,
\frac{1}{3}\,{\frac {
5\,{\it n_p}-4\,a+4\,{\it n_p}\,a}{{\it n_p}}}\right]
\end{displaymath}
since the right endpoint of this interval is positive by assumption. To determine the sign of the function, we can can compute the critical points of the derivative of the cubic equation as:
\begin{equation}
\frac{\left( {\frac {7}{9}}\,{\it n_p}+{\frac {8}{9}}\,{\it n_p}\,a-{\frac {8
}{9}}\,a\pm1/9\,\sqrt {v} \right)}{{{\it n_p}}}
\end{equation}
where:
\begin{multline*}
v = 19\,{{\it n_p}}^{2}+28\,a{{\it n_p}}^{2}-28\,{\it n_p}\,a+16\,{{\it n_p}}^{2}{a}^{2}-\\32\,{\it n_p}\,{a}^{2}+16\,{a}^{2}
\end{multline*}
The positive root is clear and the second derivative of the cubic equation in $\epsilon$ evaluated at this root is:
\begin{equation}
6\,{\it np}\,\sqrt {v} > 0
\end{equation}
meaning that the positive root corresponds to a minimum and thus, for all appropriately chosen values of $n_p$ and $a$, we know that $\dot{x} < 0$ when $\epsilon > 0$ and $x$ must decrease toward $0$. By a similar argument, we can show that if $\epsilon < 0$, then $\dot{x} > 0$ and $x$ must increase toward $1$. Thus we have proved:
\begin{lemma} Assume a non-spurious, non-trivial interior equilibrium exists in the game; i.e., the ninth equilibrium point from Theorem \ref{thm:Equil} is contained in $(0,1) \times (0,1)$. If:
\begin{equation}
x(0) < \tfrac{1}{3}{\tfrac {5{\it n_p}+4\,{\it n_p}\,a-4\,a}{{\it n_p}}}
\end{equation}
then $\lim_{t \rightarrow \infty} x(t) = 0$. Otherwise, if $x(0)$ is greater than this value, $\lim_{t \rightarrow \infty} x(t) = 1$.
\end{lemma}
To complete the characterization of the limiting behavior of the differential equation, we analyze the eigenvalues of the Jacobian matrix at $x=0,z=1$ and $x=1,z=0$. Our last lemma follows from Theorem 3.1 of \cite{Verh06}:
\begin{lemma} The point $x = 0,z = 1$ is stable if and only if:
\begin{equation}
n_p < \frac{8}{11} \quad \text{and} \quad a<\frac{5}{4}\,{\frac {{\it n_p}}{1-{\it n_p}}}
\end{equation}
Furthermore, the point $x = 1, z=0$ is stable if and only if:
\begin{equation}
n_p < \frac{4}{5} \quad \text{and} \quad \frac{1}{2}\,{\frac {{\it n_p}}{-1+{\it n_p}}}<a
\end{equation}
\end{lemma}
From these lemmas and Corollary \ref{thm:Interior}, we have the following theorem:
\begin{theorem} Suppose $n_p > \tfrac{4}{5}$ and 
\begin{displaymath}
\frac{2}{5}\,{\frac {5\,{\it n_p}-3}{1-{\it np}}} < a < \frac{2}{5}\,{\frac {3-{\it n_p}}{1-{\it np}}}
\end{displaymath}
The basin of attraction for $x=1,z=1$ is the set of $(x,z)$ pairs so that:
\begin{gather*}
\tfrac{1}{3}{\tfrac {5{\it n_p}+4\,{\it n_p}\,a-4\,a}{{\it n_p}}} < x \leq 1\\
0 < z < 1 
\end{gather*}
and the basin of attraction for $x=0,z=0$ is:
\begin{gather*}
0\leq < x < \tfrac{1}{3}{\tfrac {5{\it n_p}+4\,{\it n_p}\,a-4\,a}{{\it n_p}}}\\
0 < z < 1 
\end{gather*}
\label{thm:Basin}
\end{theorem} 
Theorem \ref{thm:Basin} is illustrated in Figure \ref{fig:Cliff1}. The more complex behaviors this system is able to exhibit yield more complex basins of attraction. However, since we anticipate $n_p > 0.9$, we have focused on this case explicitly in Theorem \ref{thm:Basin}.

This tells us that under certain conditions the final state of an (online) society (governed by these simple dynamics) can depend substantially on the initial conditions of the system. That is, if Reddit had been governed in this fashion, but the initial user group was slightly less interested in the posting of honest information (but the moderators were), then Reddit could have easily descended into a more chaotic state.

\section{Future Directions} \label{sec:Conclusion}
We discuss three future directions for research: an incentives based model, an optimal control model in which the penalty is assigned dynamically and a model in which $n_p$ is not fixed and determined by epidemic dynamics. 

\subsection{Incentives Based Model}
As noted, in the literature review, for online communities, punishments, although applied, are shown not to be truly effective \cite{SM12}. In this instance, incentives are applied to encourage good behavior. We can modify the payoff matrix $\mathbf{B}$ as:
\begin{equation}
\mathbf{B} = \begin{bmatrix}a & \frac{a}{2} \\ 0 & -1\end{bmatrix}
\label{eqn:NewB}
\end{equation}
In this scenario, we can vary $a$ to adjust the payoff received by a user engaged in cooperative behavior. The following result is immediate:
\begin{theorem} When $\mathbf{B}$ is given by Equation \ref{eqn:NewB} and the remaining payoff matrices are held constant, the new set of differential equations has four equilibria: $x=0$, $z=0$, $x=1$, $z=1$, $x=1$, $z=0$ and $x=0$, $z=1$. Furthermore, $x=1$, $z=1$ is stable just in case:
\begin{equation}
\frac{1}{2}\frac{n_p}{1-n_p} < a
\end{equation}
and $x=0$,$z=0$ is stable just in case:
\begin{equation}
a < \frac{1}{2}\frac{9n_p-4}{1-n_p}
\end{equation}
\end{theorem}
Example dynamics for the incentivizing game are shown in Figure \ref{fig:Incentives}
\begin{figure}
\centering
\includegraphics[scale=0.3]{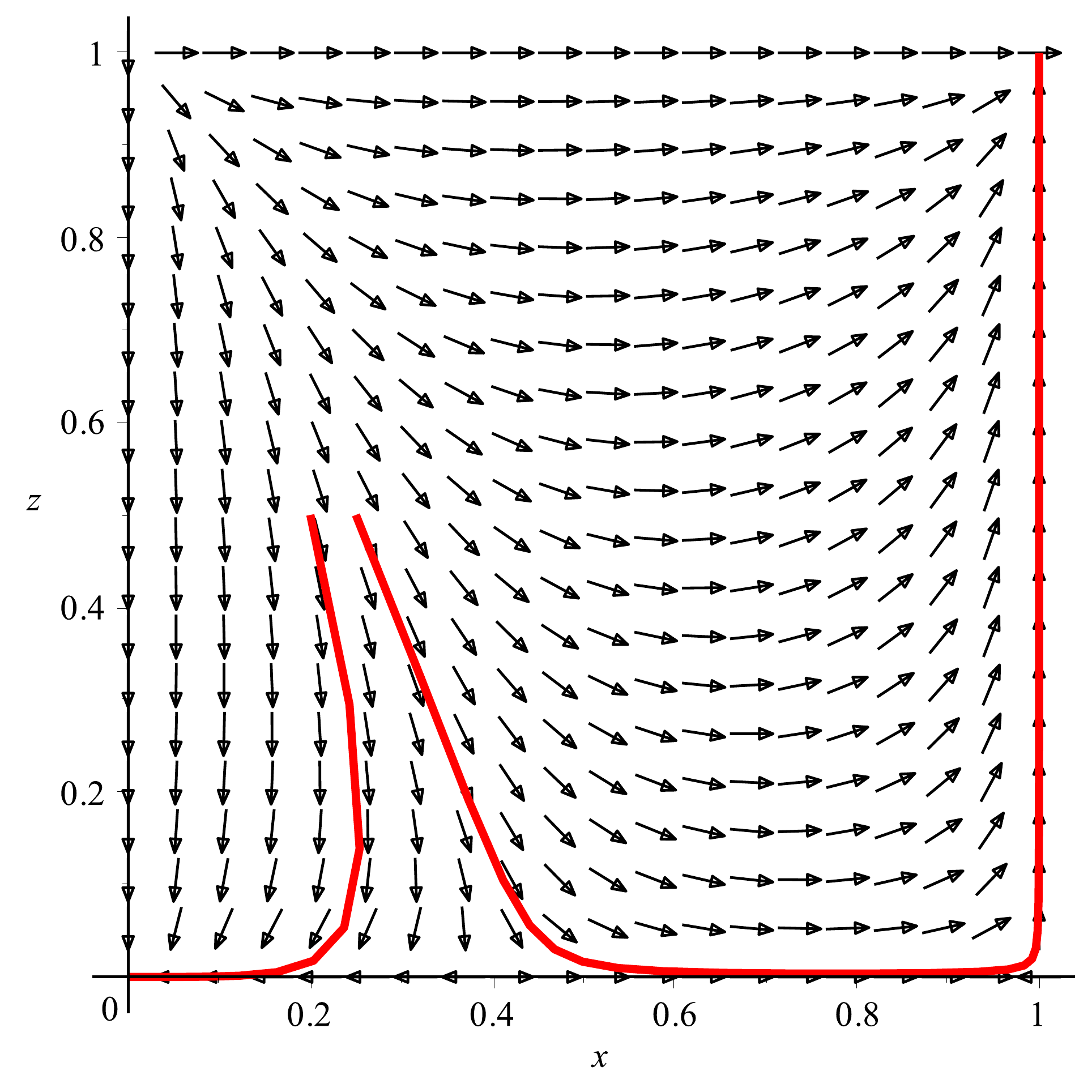}
\caption{Dynamics for the incentives game in which $a=15$ and $n_p = 0.9$.}
\label{fig:Incentives}
\end{figure}
Under these dynamics, the basin of attraction for the utopian solution $(x=1,z=1)$ is difficult to identify in closed form, however you will observe he have dynamics similar to those shown in Figure \ref{fig:Cliff2}. Additional work on this problem might yield interesting conditions on the incentive structures for encouraging stable and beneficial social networks.

\subsection{An Optimal Control Problem}
In most Social Networking sites, the number of moderators is static (that is, both $n_p$ and $n_c$ are fixed). For the remainder of this discussion, we will assume our original dynamics, rather than the incentivizing behavior describe above. If we can measure $x(0)$ (and $y(0)$), we would like to identify a time varying optimal value for $a$ (the penalty) so that $x^* = 1, z^* = 1$ is an attractor that is reachable from $x(0)$. However, an $a$ that is too large may cripple the social network (in a way not captured by the dynamics in this paper). We can phrase this problem as a finite (or infinite) time horizon optimal control problem:
\begin{equation}
\begin{aligned}
\min\;\; & \int_0^T (1-x(t))^2 + (1-z(t))^2 + a(t)^2 dt\\
s.t.\;\; & \dot{x} = \frac{1}{4}\,x \left( -1+x \right)  \left( -3{\it n_p}\,x+5\,{\it n_p}+4\,{\it n_p}\,a-4\,a \right)\\
& \dot{z} = -\frac{1}{4}z \left( -1+z \right)  \left( -16{\it n_p}\,z+16\,z+5\,{\it n_p}
-8+5{\it n_p}\,x \right) \\
& x(0) = x_0, z(0) = z_0\\
& a(t) \in [0,\infty)
\end{aligned}
\label{eqn:OptimalControl}
\end{equation}
This problem will have a Hamiltonian \cite{Ki04} that is quadratic in $a$ and thus may admit a non bang-bang solution. Study of this problem is reserved for future work. 

\subsection{Time Varying Population Proportions}
As a second generalization of this problem, consider the case where $n_p$ is not static. We can model this scenario using epidemic dynamics in which users become moderators in a manner consistent with an infection:
\begin{gather*}
\dot{n}_p = \lambda + \rho n_c-\beta n_pn_c - \mu n_p\\
\dot{n}_c = \beta n_p n_c - \rho n_c - \mu n_c 
\end{gather*}
If we assume a stable population, then $\lambda = \mu$. Here $\beta$ is the infection rate, while $\rho$ is a recovery rate that leads back to a susceptible state. Using these differential equations with the equations from (\ref{eqn:DynSysA2} - \ref{eqn:DynSysA2}) yields a more realistic dynamic. We can also define a more complex optimal control problem in which we attempt to find values for $\beta$ and $a$ that minimize the objective functional of Expression \ref{eqn:OptimalControl}.

\section{Conclusions}
In this paper we constructed a simple model of a self-regulating system that describes (to some degree) the behavior of a moderated online community. We showed that we need only two differential equations to describe a complex two-population evolutionary game, rather than the six that would be used following the analytical techniques described in \cite{Hof96}. We completely described the nature of the evolutionary dynamics of the proposed system and illustrated how the level of incentive (or penalty) associated with meeting a moderator can effect the limiting behavior of the system. These results are related to the online system Reddit. We also discuss future work in which the population structure is varied and we posed an optimal control problem that is relevant to the mechanism design for social networks.

\vspace*{5em}

\bibliographystyle{IEEEtran}
\bibliography{References,bib}

\begin{thebibliography}{10}
\providecommand{\url}[1]{#1}
\csname url@samestyle\endcsname
\providecommand{\newblock}{\relax}
\providecommand{\bibinfo}[2]{#2}
\providecommand{\BIBentrySTDinterwordspacing}{\spaceskip=0pt\relax}
\providecommand{\BIBentryALTinterwordstretchfactor}{4}
\providecommand{\BIBentryALTinterwordspacing}{\spaceskip=\fontdimen2\font plus
\BIBentryALTinterwordstretchfactor\fontdimen3\font minus
  \fontdimen4\font\relax}
\providecommand{\BIBforeignlanguage}[2]{{%
\expandafter\ifx\csname l@#1\endcsname\relax
\typeout{** WARNING: IEEEtran.bst: No hyphenation pattern has been}%
\typeout{** loaded for the language `#1'. Using the pattern for}%
\typeout{** the default language instead.}%
\else
\language=\csname l@#1\endcsname
\fi
#2}}
\providecommand{\BIBdecl}{\relax}
\BIBdecl

\bibitem{Hof96}
J.~Hofbauer, ``Evolutionary dynamics for bimatrix games: A {H}amiltonian
  system?'' \emph{J. Math. Bio}, vol.~34, pp. 675--688, 1996.

\bibitem{davis}
N.~J. Davis, ``Labeling theory in deviance research. a critique and
  reconsideration,'' \emph{The Sociological Quarterly}, vol.~13, no.~4, pp.
  447--474, 1972.

\bibitem{17}
P.~Shaver, J.~Schwartz, D.~Kirson, and C.~O'Connor, ``Emotion knowledge further
  exploration of a prototype approach,'' \emph{Journal of Personality and
  Social Psychology}, vol.~52, no.~6, 1987.

\bibitem{18}
J.~R. Suler and W.~L. Phillips, ``The bad boys of cyberspace: Deviant behavior
  in a multimedia chat community,'' \emph{CyberPsychology \& Behavior}, vol.~1,
  pp. 275--294, 1998.

\bibitem{20}
A.~Bruckman, C.~Danis, C.~Lampe, J.~Sternberg, and C.~Waldron, ``Managing
  deviant behavior in online communities,'' in \emph{CHI Extended Abstracts},
  2006, pp. 21--24.

\bibitem{22}
\BIBentryALTinterwordspacing
J.~Denegri-Knott and J.~Taylor, ``The labeling game: a conceptual exploration
  of deviance on the internet,'' \emph{Soc. Sci. Comput. Rev.}, vol.~23, no.~1,
  pp. 93--107, Mar. 2005. [Online]. Available:
  \url{http://dx.doi.org/10.1177/0894439304271541}
\BIBentrySTDinterwordspacing

\bibitem{six}
S.~David and T.~Pinch, ``Six degrees of reputation: The use and abuse of online
  review and recommendation systems,'' \emph{First Monday}, vol.~11, no.~3,
  2006.

\newpage

\bibitem{Wang:2010}
A.~H. Wang, ``Detecting spam bots in online social networking sites: a machine
  learning approach,'' in \emph{Proceedings of the 24th annual IFIP WG 11.3
  working conference on Data and applications security and privacy}, ser.
  DBSec'10.\hskip 1em plus 0.5em minus 0.4em\relax Berlin, Heidelberg:
  Springer-Verlag, 2010, pp. 335--342.

\bibitem{44}
A.~Sureka, ``Mining user comment activity for detecting forum spammers in
  youtube,'' \emph{CoRR}, vol. abs/1103.5044, 2011.

\bibitem{40}
T.~Moh and A.~Murmann, ``Can you judge a man by his friends?-enhancing spammer
  detection on the twitter microblogging platform using friends and
  followers,'' \emph{Information Systems, Technology and Management}, pp.
  210--220, 2010.

\bibitem{30}
A.~Antonucci and C.~de~Campos, ``Decision making by credal nets,'' in
  \emph{Intelligent Human-Machine Systems and Cybernetics (IHMSC), 2011
  International Conference on}, vol.~1, aug. 2011, pp. 201 --204.

\bibitem{37}
\BIBentryALTinterwordspacing
K.~Lee, J.~Caverlee, and S.~Webb, ``Uncovering social spammers: social
  honeypots + machine learning,'' in \emph{Proceedings of the 33rd
  international ACM SIGIR conference on Research and development in information
  retrieval}, ser. SIGIR '10.\hskip 1em plus 0.5em minus 0.4em\relax New York,
  NY, USA: ACM, 2010, pp. 435--442. [Online]. Available:
  \url{http://doi.acm.org/10.1145/1835449.1835522}
\BIBentrySTDinterwordspacing

\bibitem{ot2}
Fassim, ``Fassim: a forum spam prevention plugin,''
  http://www.fassim.com/about/.

\bibitem{onlinetool}
S.~F. SPam, 2012, http://www.stopforumspam.com.

\bibitem{West:2010}
\BIBentryALTinterwordspacing
A.~G. West, S.~Kannan, and I.~Lee, ``Stiki: an anti-vandalism tool for
  wikipedia using spatio-temporal analysis of revision metadata,'' in \emph{6th
  International Symposium on Wikis and Open Collaboration}, ser. WikiSym
  '10.\hskip 1em plus 0.5em minus 0.4em\relax New York, NY, USA: ACM, 2010, pp.
  32:1--32:2. [Online]. Available:
  \url{http://doi.acm.org/10.1145/1832772.1832814}
\BIBentrySTDinterwordspacing

\bibitem{feldman2006}
M.~Feldman, ``Free-riding and whitewashing in peer-to-peer systems,''
  \emph{IEEE Journal on Selected Areas in Communication}, vol.~24, no.~5, pp.
  1010 -- 1019, 2006.

\bibitem{JCC4:JCC401}
\BIBentryALTinterwordspacing
J.~Boyd, ``In community we trust: Online security communication at ebay,''
  \emph{Journal of Computer-Mediated Communication}, vol.~7, no.~3, pp. 0--0,
  2002. [Online]. Available:
  \url{http://dx.doi.org/10.1111/j.1083-6101.2002.tb00147.x}
\BIBentrySTDinterwordspacing

\bibitem{Wei95}
J.~W. Weibull, \emph{Evolutionary Game Theory}.\hskip 1em plus 0.5em minus
  0.4em\relax MIT Press, 1997.

\bibitem{Gri12}
C.~Griffin, ``Graph theory: Penn state math 485 lecture notes (v 0.9.8),''
  http://www.personal.psu.edu/cxg286/Math485.pdf, 2011-2012.

\bibitem{Verh06}
F.~Verhulst, \emph{Nonlinear Differential Equations and Dynamical Systems},
  2nd~ed.\hskip 1em plus 0.5em minus 0.4em\relax Springer, 2006.

\bibitem{SM12}
G.~P. Anna~Squicciarini, William~Mcgill and S.~Huang, ``Early detection of
  policies violations in a social media site: A bayesian belief network
  approach,'' in \emph{IEEE Symposium on Policies for Distributed Systems \&
  Networks}, 2012.

\bibitem{Ki04}
D.~E. Kirk, \emph{Optimal Control Theory: An Introduction}.\hskip 1em plus
  0.5em minus 0.4em\relax Dover Press, 2004.

\end{thebibliography}

\end{document}